%% file: NetOrientations.tex
\newif\ifConference
\newif\ifJournal
\newif\ifAnonymous
\newif\ifFinal

\Conferencetrue
\Finaltrue

\ifConference
\documentclass[envcountsect,envcountsame]{llncs}
\else
\documentclass[runningheads,envcountsect,envcountsame]{llncs}
\fi

\input{commands.tex}

\ifAnonymous

	\author{Anonymous Author(s)}
	\authorrunning{Anonymous Author(s)}
	\institute{Anonymous Affiliation(s)}

\else

	\author{
		Jannik~Schestag\inst{1}\orcidID{0000-0001-7767-2970}\thanks{
			Funded by ERC Advanced Grant, New Horizons of Parameterized Complexity (NewPC), Grant ID 101199930.
		} \and
		Norbert~Zeh\inst{2}\orcidID{0000-0002-0562-1629}\thanks{
			Research supported by NSERC Discovery Grant number RGPIN-2025-06235.
		}
	}
	\institute{
		University of Bergen, Norway;
		\email{j.schestag@uib.no};
		\url{www.JannikSchestag.eu}\and
		Dalhousie University, Halifax, Canada;
		\email{nzeh@cs.dal.ca}
	}

\fi

\Crefname{rrule}{Rule}{Rules}
\Crefname{observation}{Observation}{Observations}
\Crefname{construction}{Construction}{Constructions}

\usetikzlibrary{shadows}

\newcommand{\TTsO}{\PROB{Topological Simultaneity Ordering}}
\newcommand{\TTsOshort}{\PROB{TSO}}

\newcommand\C{\mathcal{C}}
\newcommand\N{\mathcal{N}}

\newcommand\U{\mathcal{U}}
\newcommand\Z{\mathcal{Z}}
\newcommand\gen[1]{#1_{\textrm{gen}}}

\let\nf\nicefrac

\newcommand{\todosi}[1]{\todo[inline,color=red!25!green!50]{ #1}}

\newcommand{\appendixstar}{\ensuremath{\bigstar}}

\title{Orienting Unrooted Binary Networks Faster: Focus on the Generator}

\titlerunning{Faster Network Orientation}

\begin{document}

\maketitle

\begin{abstract}
  The problem of orienting an unrooted phylogenetic network to obtain a
  rooted network belonging to a specific network class is NP-hard in many cases.
  In this paper, we introduce two algorithmic frameworks that yield
  significantly improved fixed-parameter tractable (FPT) algorithms
  parameterized by the network level~$\ell$.
  We show that, for several prominent network classes, the core algorithmic
  difficulty lies in finding a directed spanning tree of the network's
  unrooted generator.
  By enumerating these spanning trees in $\Oh(5.3334^\ell)$ time and
  orienting all remaining edges in polynomial time, we solve the orientation
  problem in $\Oh(5.3334^\ell \cdot n)$ time for tree-based networks and in
  $\Oh(5.3334^\ell \cdot n^2)$ time for orchards, where~$n$ is the number of
  vertices of the network.
  Extending this approach with further branching yields $\Oh(10.6667^\ell \cdot
  n^2)$-time algorithms for tree-child and normal networks.
  Our second technique bypasses spanning trees by directly guessing the
  placement of reticulations on the generator.
  This framework provides $\Oh(12.2071^\ell \cdot n^2)$-time algorithms for
  temporal, reticulation-visible, and tree-sibling networks.
  We further demonstrate the versatility of the reticulation guessing
  framework by showing that even computing an orientation with minimum
  scanwidth is single-exponential FPT with respect to the level.
\end{abstract}

\section{Introduction}

\label{sec:intro}

Phylogenetic networks are a crucial tool in evolutionary biology that extends
traditional phylogenetic trees to evolutionary histories that may
also contain reticulate evolutionary events such as hybridization, horizontal
gene transfer, and recombination.
In practice, many network reconstruction methods infer unrooted networks,
leaving the evolutionary direction of the edges unresolved.
To interpret these networks biologically, it is often necessary to orient the edges
to form a valid directed acyclic graph with a single root.
However, not all orientations produce biologically meaningful networks,
motivating the problem of finding an orientation that belongs to a well-studied
class of phylogenetic networks, such as tree-child, tree-based, or orchard
networks.

This problem has received significant attention in recent years.
While a binary unrooted network can be oriented to obtain a valid
rooted network in linear time~\cite{bulteau2023turning}, restricting the output
to belong to some network class with specific properties drastically increases
the computational complexity.
The orientation problem has been shown to be NP-hard for tree-based
networks~\cite{huber2024orienting}, orchards~\cite{dempsey2024wild}, and
tree-child networks~\cite{docker2025existence,iersel2026class}.

Heuristic and exact approaches for orienting networks have been explored in
Verzijlbergen's thesis~\cite{VerzijlbergenThesis}.
Further mathematical conditions for a network to be orientable into a
tree-child network have been analyzed in~\cite{maeda2023orienting}.
Parameterized approaches have gained popularity in this field.
In~\cite{Ortientation}, the closely related problem of finding
degree-constrained acyclic orientations has been studied.
In~\cite{huber2024orienting}, a general framework for network orientation has been
presented, showing that orienting into a network class~$\C$ is
fixed-parameter tractable (FPT) with respect to the level of the network,
provided the class~$\C$ satisfies three mild conditions.
The level of a network is an important parameter in the study of phylogenetics,
as it measures the number of reticulation events in blobs (i.e., reticulate
components).
Though demonstrating that the orientation problem for many network classes is FPT,
the running time achieved by the algorithm in \cite{huber2024orienting} is
impractical.
The paper presents rather loose bounds that depend on the specific network
class but, for many network classes, would exceed $\Oh^*(1024^\ell)$.
A more careful analysis seems to yield running times around $\Oh^*(28.6^\ell)$
for some network classes.
Recently, an impressive FPT algorithm for orienting unrooted networks that is
fast in practice has been proposed in~\cite{urata2026orientability}, though
with a larger parameter.

This paper presents two frameworks that significantly improve
the best known running times for orienting binary networks into various
network classes.
Both of our approaches leverage network generators \cite{kelk-linz-tbr}.

We show that, for tree-based networks and orchards, the core of the orientation
problem is to find the correct directed spanning tree of the generator.
Since the generator of a level-$\ell$ network has at most $2\ell$ vertices, we
can enumerate its spanning trees in $\Oh(5.3334^\ell)$ time.
Given the correct spanning tree, the remainder of the problem is solvable in
polynomial time for tree-based and orchard networks, yielding $\Oh(5.3334^\ell
\cdot n)$-time and $\Oh(5.3334^\ell \cdot n^2)$-time algorithms for these
classes, respectively.
By making further guesses on paths in the network corresponding to
non-tree edges of the generator, we extend this technique to solve the
orientation problem for tree-child and normal networks in $\Oh(10.6667^\ell
\cdot n^2)$ time.

Similar to the approch taken in \cite{huber2024orienting}, our second approach
guesses the locations of reticulations but much less precisely than in
\cite{huber2024orienting}.
We believe this technique to be powerful enough to solve the orientation
problem for most network classes and demonstrate its usefulness by obtaining
$\Oh(12.2071^\ell \cdot n^2)$-time algorithms for temporal,
reticulation-visible, and tree-sibling networks.
We also demonstrate that this approach can be used to prove that computing a
network orientation with minimum scanwidth is FPT in the level.

\ifConference
  Proofs of propositions marked with {\appendixstar} are partly or completely omitted
and can be found in the full paper.
\fi

\section{Preliminaries}

\label{sec:prelims}

We use the notations $[n] := \{1,\dots,n\}$ and $[n]_0 := \{0,\ldots,n\}$.
The $\Oh^*$ and $\tilde\Oh$ notations suppress polynomial and poly-logarithmic
factors in the input size, respectively.

We use standard graph-theoretic notation as in~\cite{diestel} and define only
some terms here.
We consider finite and simple graphs~$G$ with vertex set~$V(G)$ and edge
set~$E(G)$.
An \emph{unrooted (binary phylogenetic) network} is an undirected graph $\U$
whose vertices have degree $3$ (\emph{internal vertices}) or $1$
(\emph{leaves}).
We write~$\{u,v\}$ for undirected edges and~$uv$ for directed edges, which we
call \emph{arcs}.
To \emph{suppress} a degree-$2$ vertex $v$ in an undirected graph is to
replace $v$ and its two incident edges $\{u,v\}$ and $\{v,w\}$ with an edge
$\{u,w\}$.

A \emph{spanning tree} of a connected undirected graph~$G$ is a connected
acyclic subgraph that contains all vertices of~$G$.
A \emph{topological ordering} of a directed acyclic graph is a linear ordering
$<$ of its vertices such that \mbox{$u < v$} for every arc~$uv$.
A~graph is \emph{bipartite} if its vertex set can be partitioned into two
independent sets $V_1$ and $V_2$.
A~\emph{matching} is a set of edges sharing no endpoints.

A \emph{$2$-edge-connected component} of an undirected graph is a maximal
subgraph that remains connected after the removal of any single edge.
An edge whose removal increases the number of connected components in the graph
is a \emph{cut edge}.
A \emph{blob} of an unrooted network $\U$ is a subgraph induced by the
vertices in a $2$-edge-connected component of $\U$ that is not a leaf of $\U$
and all vertices adjacent to this component.
A network is \emph{simple} if it has a single blob (which is the whole
network).
Every blob of a network $\U$ is itself a simple network.

A \emph{rooted (binary phylogenetic) network} is a directed acyclic graph with
a single source, called its \emph{root}, which has out-degree $2$; all other
vertices have in-degree~1 and out-degree~0 (\emph{leaves}), in-degree~1 and
out-degree~2 (\emph{tree vertices}) or in-degree~2 and out-degree~1
(\emph{reticulations}).
In biological applications, leaves are usually associated with species.
The \emph{reticulation number} of a rooted binary network $\N$ is the number of
reticulations, which is equal to $|E(\N)| - |V(\N)| + 1$.
We also define the reticulation number of an unrooted network $\U$ to be
$|E(\U)| - |V(\U)| + 1$.
The \emph{level} of a network is the maximum reticulation number of its
blobs.

An unrooted network~$\U$ can be \emph{oriented into} a rooted network~$\Net$
if the underlying undirected graph of~\Net can be obtained from $\U$ by
subdividing a single edge (the introduced vertex is the root of $\N$).
Given a class of rooted networks~$\C$, we say an unrooted network~$\U$ can be
\emph{$\mathcal{C}$-oriented} if~$\U$ can be oriented into a rooted network $\N
\in \C$.
The main problem studied in this paper is this:

\problemdef
{\textsc{$\mathcal{C}$-Orientation}}
{An unrooted binary network~$\U$.}
{A $\mathcal{C}$-orientation of $\U$ if it exists, and NO otherwise.}

\looseness=-1
We define specific network classes as needed later in this paper.
See~\cite{Hellmuth2023,Kong2022} for their biological interpretations.
A network class $\C$ is \emph{blob-determined} if a network $\N$ belongs to
$\C$ if and only if its blobs belong to $\C$ \cite{huber2024orienting}.
All network classes studied in this paper are blob-determined.
Thus, by the following lemma, it suffices to show how to efficiently find
a $\C$-orientation for any given blob.

\begin{restatable}[\appendixstar]{lemma}{blobdetermined}\label{lem:blob-determined}
  If\/ $\C$ is a blob-determined network class and the $\C$-orientation problem
  can be solved in $\Oh(f(\ell) \cdot |B|)$ time, for any level-$\ell$ blob $B$
  and any possible choice of the root of $B$, then the $\C$-orientation problem
  can be solved in $\Oh(f(\ell) \cdot n^2)$ time, for any unrooted level-$\ell$
  network.
\end{restatable}

\thmtoappendix{lem:blob-determined}{\blobdetermined*}{
\begin{proof}
  We start by guessing the edge of $\U$ to be subdivided to place the root
  $\rho$ of the orientation.
  As there are $\Oh(n)$ edges in $\U$, there are $\Oh(n)$ choices.
  For each choice of placing $\rho$, the root $\rho_B$ of every blob $B$ in the
  orientation is fixed by this choice: For the unique blob $B$ that contains
  $\rho$, we have $\rho_B = B$.
  For any other blob $B$, $\rho_B$ is the unique vertex in $B$ that belongs to
  every path from $\rho$ to any vertex in the $2$-edge-connected component
  corresponding to $B$.
  Thus, an orientation with root $\rho$ can be obtained by independently
  finding an orientation of each blob $B$ with root $\rho_B$.
  Since this takes $\Oh(f(\ell) \cdot |B|)$ time per blob, the total running
  time of the algorithm is $\Oh(n \cdot \sum_B f(\ell) \cdot |B|) = \Oh(f(\ell)
  \cdot n^2)$.\qed
\end{proof}
}

\section{Unrooted Generators and Enumerating Spanning Trees}

\label{sec:generators}

The \emph{generator} $\gen{\U}$ of a simple unrooted network $\U$ is the
graph obtained by removing leaves and suppressing degree-$2$ vertices until
only degree-$3$ vertices remain~\cite{kelk-linz-tbr}.
To avoid confusion with the edges of $\U$, we call the edges of $\gen{\U}$
\emph{sides}.
Each side $e = \{u,w\}$ of $\gen{\U}$ corresponds to a path~$P_e = \langle u,
v_1, \dots, v_{k(e)}, w \rangle$ of length~$k(e)+1$ (counting the number of
edges), each of whose internal vertices~$v_i$ is adjacent to a leaf~$x_i$.
We call $e$ a \emph{$k(e)$-side} or a~\emph{$\geq q$-side} for~$q \leq k(e)$ to
specify the number of pendant leaves of~$P_e$.

For most network classes $\C$, our algorithm for finding a $\C$-orientation
starts by guessing which edge of $\U$ is subdivided to introduce the root
$\rho$ of the orientation.
For every $2$-edge-connected component $C$ of $\U$ with corresponding blob $B$,
this fixes the root $\rho_B$ as the unique vertex in $C$ that is on every path
from $\rho$ to a vertex in $C$, provided $B$ does not contain $\rho$.
For the blob $B$ that contains $\rho$, we set $\rho_B = \rho$.
If we choose the root placement before partitioning $\U$ into its blobs, we
exclude the cut edge incident to $\rho_B$ from every blob $B$.
This makes $\rho_B$ the unique degree-$2$ vertex in each blob $B$.
When constructing the generator~$\gen{B}$ of~$B$, this is the only degree-$2$
vertex we \emph{do not} suppress.

\begin{observation}
	\label{obs:comp-undir-generator}
	The generator of an unrooted network can be computed in linear time.
\end{observation}

\begin{lemma}[Lemma~1 in \cite{kelk-linz-tbr}]
	\label{lem:size-gen}
	The generator of a level-$\ell$ blob has $2\ell - 1$ vertices and $3\ell - 2$ edges,
	accounting for the degree-2 root vertex.
\end{lemma}

The algorithms in \cref{sec:spanning-trees} need to ``guess'' the correct
spanning tree of each blob $B$, that is, they enumerate all possible spanning
trees and, for each, test whether there is an orientation of $B$ consistent
with the chosen root and the chosen spanning tree.
Thus, we need a bound on the number of spanning trees of a generator and on
the time required to enumerate them.

\begin{restatable}[\appendixstar]{lemma}{spanning}\label{obs:spanning}
	The generator of a level-$\ell$ blob has $\Oh((\nf{16}{3})^\ell)$ spanning
	trees, which can be enumerated in $\Oh((\nf{16}{3})^\ell)$ time.
\end{restatable}

\thmtoappendix{obs:spanning}{\spanning*}{
\begin{proof}
	As shown in \cite{mckay1983spanning}, the number of spanning trees of a
	$\Delta$-regular multigraph of order~$n$ converges, as $n \to \infty$, to
	\[
	\left(\frac{(\Delta-1)^{\Delta-1}}{\left(\Delta^2 - 2\Delta\right)^{\Delta/2 - 1}}\right)^n.
	\]
	The generator of a level-$\ell$ blob is $3$-regular and, by
	Lemma~1 in \cite{kelk-linz-tbr}, has $2\ell - 2$ vertices.  Thus, the
	number of spanning trees is $\Oh((\nf{16}{3})^\ell)$.
	
	Given a graph with $n$ vertices and $m$ edges, all~$T$~spanning trees can be
	enumerated in $\Oh(n + m + T)$
	time~\cite{kapoor1995algorithms,shioura1997optimal}.
	Thus, enumerating all spanning trees of the generator of a level-$\ell$ blob
	takes $\Oh((\nf{16}{3})^\ell)$ time. \qed
\end{proof}
}

\section{Finding Orientations by Guessing the Spanning Tree}

\label{sec:spanning-trees}

\subsection{Tree-Based Networks}

\label{sec:tb}

A network $\Net$ is \emph{tree-based} if it has a spanning tree whose leaves
are leaves of $\N$.
We call such a spanning tree a \emph{base tree} of $\Net$.
This definition is independent of edge directions, yielding the following lemma.

\begin{restatable}[\appendixstar]{lemma}{unrootedbasetree}\label{lem:un-rooted-base-tree}
	An unrooted network has a tree-based orientation if and only if it
	has an unrooted base tree.
\end{restatable}

\thmtoappendix{lem:un-rooted-base-tree}{\unrootedbasetree*}{
\begin{proof}
  Given an unrooted network $\U$ and a tree-based orientation $\N$ of $\U$,
  discarding directions and suppressing the root of a rooted base tree of $\N$
  yields an unrooted base tree of~$\U$.

  Conversely, if $\U$ has an unrooted base tree $T$, we obtain a directed
  tree $T^r$ by subdividing one of its edges to introduce a root $\rho$, and
  directing all edges away from $\rho$.
  Subdividing the same edge in $\U$ yields an unrooted graph $\U^r$.
  We direct the edges of $\U^r \cap T^r$ as in $T^r$, and direct all remaining
  edges of $\U^r$ according to a topological ordering of $T^r$.
  This makes $\U^r$ a rooted network with directed base tree $T^r$.\qed
\end{proof}
}


By \cref{lem:un-rooted-base-tree}, deciding tree-based orientability reduces to
finding an unrooted base tree.
Moreover, as the class of tree-based networks is blob-determined, it suffices
to find an unrooted base tree for each blob $B$.
The union of these trees is a base tree of $\U$.
Thus, to prove the following theorem, it suffices to decide
in~$\Oh((\nf{16}{3})^\ell \cdot |B|)$~time whether each blob $B$ of $\U$ has an
unrooted base tree.

\begin{theorem}
	\label{thm:tb}
	\PROB{Tree-Based Orientation} can be solved in $\Oh(5.3334^\ell \cdot n)$ time.
\end{theorem}

Consider a fixed blob $B$ of $\U$.
For every base tree $T$ of $B$, there is a unique spanning tree $S$ of
$\gen{B}$ \emph{consistent} with~$T$ in the sense that a side $e$ of $\gen{B}$
belongs to $S$ if and only if all edges in $P_e$ belong to $T$.
Thus, we enumerate all spanning trees~$S$ of~$\gen{B}$ and test whether $B$ has
a base tree $T$ consistent with~$S$.
By \cref{obs:spanning}, $B$ has $\Oh((\nf{16}{3})^\ell) = \Oh(5.3334^\ell)$
spanning trees.
Thus, to prove \cref{thm:tb}, it suffices to show how to find, for a fixed
spanning tree $S$ of $\gen{B}$, a base tree $T$ of $B$ consistent with $S$ in
$\Oh(|B|)$ time.

Let $V_0$ be the set of leaves of $S$, let $V_1 \subseteq V_0$ be the set of
leaves of $S$ without an incident~$\ge 2$-side not in $S$, and let $E_1$ be the
set of $1$-sides not in $S$ incident to vertices in $V_1$.
Let $H$ be the bipartite graph with vertex set $V_1 \uplus E_1$ and with an edge
$(v,e)$, where $v \in V_1$ and $e \in E_1$, if $e$ is incident to $v$ in $B$.
Then every vertex of $H$ has degree at most~$2$.
Thus, a maximum matching $M$ in $H$ can be found in $\Oh(|H|) = \Oh(|B|)$ time.
Next, we show that $B$ has a base tree $T$ consistent with $S$ if and only if
$M$ matches all vertices in $V_1$, and that $T$ can be constructed from $S$ and
$M$ in linear time.
This proves \cref{thm:tb}.

If $M$ matches all vertices in $V_1$, then we construct a base tree $T$ of $B$
by taking the union $T'$ of paths corresponding to all sides in $S$, and their
pendant leaves of~$B$, adding the path $\langle v,w,x \rangle$ for each edge
$\{v,e\} \in M$, where $P_e = \langle v,w,u \rangle$ and $x$ is the leaf
adjacent to $w$, adding a path connecting each vertex $v \in V_0 \setminus
V_1$ to the closest leaf on an incident $\ge 2$-side, and greedily completing
$T'$ to a spanning tree of $B$.
This clearly takes linear time and produces a base tree of $B$.

Conversely, if $B$ has a base tree $T$ consistent with~$S$, then $T$
contains a path~$P_v$ from each vertex $v \in V_1$ to a distinct leaf of~$B$
attached to an incident side~$e \in E_1$.
The set $M = \{\{v,e\} \mid P_v \subseteq e\}$ is a matching in $H$
because no two paths~$P_u$ and~$P_v$ can end at the same leaf.
Clearly, $M$ matches all vertices in $H$.

\subsection{Orchards}

\label{sec:orchards}

An \emph{orchard} is a rooted tree-based network $\Net$ that has a base
tree~$T$ and a labelling $\lambda : V(\Net) \to \mathbb{Q}$ such that
$\lambda(u) < \lambda(v)$ for each arc $uv \in T$, and $\lambda(u) =
\lambda(v)$ for each arc $uv \notin T$ \cite{vanIersel-hgt-consistent}.
Such a labelling $\lambda$ is called an \emph{HGT-consistent labelling}
of~$\Net$.
By \cref{lem:blob-determined}, the following theorem follows if we can show
that it takes~$\Oh(5.3334^\ell \cdot |B|)$~time to decide whether any blob $B$
has an orchard orientation with a given root $\rho_B$.

\begin{theorem}
	\label{thm:orchards}
	\PROB{Orchard Orientation} can be solved in $\Oh(5.3334^\ell \cdot n^2)$ time.
\end{theorem}

To decide whether a blob $B$ has an orchard orientation with root $\rho_B$, we
guess an unrooted spanning tree $S'$ of $\gen{B}$ and obtain a directed
spanning tree~$S$ by orienting~$S'$ away from $\rho_B$.
As there are $\Oh(5.3334^\ell)$ possible spanning trees of~$B$, it now suffices
to decide in $O(|B|)$ time whether there is a directed base tree $T$ of $B$
that is consistent with $S$ and has an HGT-consistent labelling.
We reduce this to the following problem:

\problemdef
{\TTsO (\TTsOshort)}
{An undirected graph $G = (V, E)$ with $V = V_1 \uplus V_2$ and $E = E_= \uplus E_{\ne}$,
	and a set of arcs $A \subseteq (V_1 \cup V_2)^2$.}
{A labelling $\lambda: V_1 \cup V_2 \to \mathbb{Q}$ such that
	{\begin{enumerate}[(1)]
			\item~$\lambda(v) = \lambda(w)$ for each edge $\{v,w\} \in E_=$,
			\item~$\lambda(v) \neq \lambda(w)$ for each edge $\{v,w\} \in E_{\ne}$,
			\item~$\lambda(v) < \lambda(w)$ for each arc $vw \in A$, and
      \item~for each $v \in V_1$, there is an incident edge $\{v,w\} \in
      E_{\ne}$ with $\lambda(v) < \lambda(w)$.
	\end{enumerate}}
}

Specifically, we construct a \TTsOshort instance from $S$ in which $V_1$ is the
set of leaves of $S$ without an incident $\ge 2$-side not in $S$, $V_2 = V(\gen{B})
\setminus V_1$, $A$ is the set of arcs in $S$, $E_=$ is the set of $0$-sides of
$\gen{B}$ not in $S$, and $E_{\ne}$ is the set of $1$-sides of $\gen{B}$ not in
$S$.

\begin{lemma}
  \label{lem:ttso-reduction}
  $B$ has a directed base tree $T$ consistent with $S$ and an HGT-con\-sis\-tent
  labelling of $T$ if and only if the corresponding \TTsOshort instance has a
  solution.
\end{lemma}

\begin{proof}
  First, assume that the \TTsOshort instance has a solution $\lambda$.
  Then we extend $\lambda$ to a labelling of all vertices in $B$ as follows:
  For every path $P_e = \langle u, v_1, \ldots, v_k, w \rangle$ corresponding
  to an arc $e = (u,v) \in A$, we set $\lambda(v_i) = \lambda(u) + i / (k + 1)$
  for all~$i \in [k]$.
  For each side $e = \{u,w\} \not\in S$, the internal vertices of $P_e = \langle u,
  v_1, \ldots, v_k, w \rangle$ are labelled as follows: (i) if $k(e) = 0$, then
  no internal vertices need labelling; (ii)~if $k(e) = 1$, then $e \in E_{\ne}$,
  and we set $\lambda(v_1) = \max(\lambda(u), \lambda(w))$; (iii) if $k(e) \ge
  2$, then we set $\lambda(v_i) = \max(\lambda(u), \lambda(w)) + i$ for all $i
  \in[k-1]$, and $\lambda(v_k) = \lambda(v_{k-1})$.
  For every leaf~$x$ of $B$ with parent $p$ in $B$, we set $\lambda(x) = \lambda(p) +
  1$.
  We obtain the final orientation by directing all edges from vertices with
  lower labels to vertices with higher labels, breaking ties arbitrarily.
  It is easily verified that every non-leaf vertex $v$ of $B$ has a neighbour $w$ with
  $\lambda(v) < \lambda(w)$, and that every non-root vertex $v$ has exactly one
  neighbour $u$ with $\lambda(u) < \lambda(v)$.
  Thus, the arcs $xy$ with $\lambda(x) < \lambda(y)$ in the orientation form a
  base tree.
  Since each remaining arc $uv$ satisfies $\lambda(u) = \lambda(v)$, $\lambda$
  is an HGT-consistent labelling of $B$ with respect to this base tree, that is,
  the orientation is an orchard.

  Conversely, if $B$ has an orchard orientation with root $\rho_B$, base tree
  $T$ consistent with $S$, and HGT-consistent labelling~$\lambda$, then
  $\lambda(u) < \lambda(v)$ for each arc $e = uv \in S$, as the corresponding path
  $P_e = \langle u, v_1, \ldots, v_k, w \rangle$ satisfies
  $\lambda(u) < \lambda(v_1) < \cdots < \lambda(v_k) < \lambda(w)$.
  Every side $e = \{u,v\} \in E_=$ is a $0$-side not in $S$ and, thus, satisfies
  $\lambda(u) = \lambda(v)$.
  Every side $e \in E_{\ne}$ is a $1$-side with corresponding path $P_e =
  \langle u, w, v \rangle$.
  Exactly one of the edges in $P_e$ does not belong to $T$.
  Thus, we have either $\lambda(u) = \lambda(w) > \lambda(v)$ or $\lambda(u) <
  \lambda(w) = \lambda(v)$.
  Finally, the child $w$ in $T$ of each vertex $u \in V_1$ must belong to a path $P_e$
  corresponding to a $1$-side $e = \{u, v\} \notin S$ incident to $u$.
  This side belongs to $E_{\ne}$ and, since $\lambda(u) < \lambda(w)$, satisfies
  $\lambda(u) < \lambda(v)$.
  This shows that the restriction~$\lambda$ to~$V(\gen{B})$ yields a solution
  of the $\TTsOshort$ instance obtained from $B$ and $S$.\qed
\end{proof}

As the \TTsOshort instance obtained from $B$ and $S$ is easily seen to have size
$\Oh(|B|)$, it suffices to show how to solve it in linear time.

\begin{restatable}[\appendixstar]{lemma}{solveTTso}\label{lem:solve-TTsO}
	Finding a solution for any instance $(V_1 \uplus V_2, E_= \uplus E_{\ne}, A)$
	of \TTsOshort, or deciding that none exists, takes $\Oh(|V_1| + |V_2| + |E_=|
	+ |E_{\ne}| + |A|)$ time.
\end{restatable}

\thmtoappendix{lem:solve-TTsO}{\solveTTso*}{
\begin{proof}
	We use the following algorithm to find a solution $\lambda$ or to decide that
	none exists.
	Let $C_1,\dots,C_k$ be the connected components of $(V_1 \uplus V_2, E_=)$.
	If there is an arc $uv$ or an edge $\{u,v\} \in E_{\ne}$ with $u,v \in C_j$, for some $j \in [k]$,
	then there is no solution.
	Otherwise, let $\Z := \emptyset$ and $h := k$.
	While $\Z \ne \{C_1, \ldots, C_k\}$, choose a component $C_i \notin \Z$ such
	that every arc $uv \in A$ with $u \in C_i$ satisfies $v \in C_j$ for some
	component $C_j \in \Z$, and every vertex $u \in V(C_i) \cap V_1$ has an
	incident edge $\{u,v\} \in E_{\ne}$ such that $v \in C_j$ for some component
	$C_j \in \Z$.
	We label the vertices in $C_i$ with $h$, add $C_i$ to $\Z$, decrement $h$,
	and repeat.
	If there is no such component and not all vertices have been labelled yet,
	then report that there is no solution.
	
	\proofpara{Correctness}
	First assume that the algorithm reports a labelling $\lambda$.
	Conditions (1) and (2) in the definition of \TTsOshort are clearly satisfied
	because we ensure that all vertices in the same component receive the same
	label, vertices in different components receive different labels, and the
	algorithm fails if there is an edge in $E_{\ne}$ with both endpoints in the
	same component.
	Conditions (3) and (4) are easily seen to be satisfied based on the condition
	for labelling the vertices in a component and adding that component to $\Z$.
	
	Now, assume that there is a solution $\lambda$ of \TTsOshort on the given input.
	Then the vertices in each connected component of~$(V,E_=)$ have the same
	label, and there are no edges in~$E_{\ne}$ or arcs in~$A$ between vertices in
	the same component.
	Moreover, the labelling $\lambda$ defines an ordering $\langle C_1, \ldots,
	C_k \rangle$ of the components of $(V,E_=)$.
	Adding components to $\Z$ in the order $\langle C_k, \ldots, C_1 \rangle$
	ensures that, when we add $C_i$ to $\Z$, we have $\Z = \{C_{i+1}, \ldots,
	C_k\}$.
	For every arc $uv$ with $u \in C_i$, we have $v \in C_j$ with $j > i$.
	For every vertex $u \in V(C_i) \cap V_1$, there is an edge $\{u,v\} \in
	E_{\ne}$ with $\lambda(u) < \lambda(v)$, so $v \in C_j$ with $j > i$.
	Thus, $C_i$ can be added to $\Z$ after adding $C_{i+1}, \ldots, C_k$ to $\Z$.
	The algorithm may add components to $\Z$ in a different order.
	However, as long as $\Z \ne \{C_1, \ldots, C_k\}$, there is an index $i$
	such that $C_i \notin \Z$ and $\Z \supseteq \{C_{i+1}, \ldots, C_k\}$.
	Thus, the algorithm can always add a component to $\Z$ and succeeds in
	finding a solution.
	
	\proofpara{Running time}
	Computing the connected components of $(V, E_=)$ takes $\Oh(|V| + |E_=|)$
	time.
	The initial check for arcs or edges in $E_{\ne}$ with both endpoints in the
	same component takes $\Oh(|E_{\ne}| + |A|)$ time.
	To carry out the rest of the algorithm in linear time, we label
	each component $C_i$ with the number of arcs $uv$ such that $u \in C_i$ and
	with the number of vertices $u \in V(C_i) \cap V_1$.
	Every time we add a component $C_j$ to $\Z$, we iterate over the arcs $uv$
	with $v \in C_j$ and, for each, decrease the arc count for the component
	$C_i$ that contains $u$.
	We also inspect all edges $\{u,v\} \in E_{\ne}$ such that $v \in C_j$ and $u
	\in V_1$, and decrease the $V_1$-vertex count of the component containing $u$
	unless we have already discovered an edge $\{u,w\} \in E_{\ne}$ with $w \in
	C_k$ and $C_k \in \Z$, which can easily be checked by maintaining a flag for
	each vertex in $V_1$.
	Using these counts, a component can be added to $\Z$ exactly when both its
	counts become $0$.
	To maintain these counts, we inspect every arc and every edge in $E_{\ne}$
	exactly once.
	Thus, this takes linear time.\qed
\end{proof}
}

\begin{sketch}
  Let $C_1,\dots,C_k$ be the connected components of $(V_1 \uplus V_2, E_=)$.
  If there is an arc $uv$ or an edge $\{u,v\} \in E_{\ne}$ with $u,v \in
  C_j$, for some $j \in [k]$, then no solution exists.
  Otherwise, let $\Z := \emptyset$ and $h := k$.
  While $\Z \ne \{C_1, \ldots, C_k\}$, choose a component $C_i \notin \Z$ such
  that every arc $uv \in A$ with $u \in C_i$ satisfies $v \in C_j$ for some
  component $C_j \in \Z$, and every vertex $u \in V(C_i) \cap V_1$ has an
  incident edge $\{u,v\} \in E_{\ne}$ such that $v \in C_j$ for some component
  $C_j \in \Z$.
  We label the vertices in $C_i$ with $h$, add $C_i$ to $\Z$, decrement $h$,
  and repeat.
  If there is no such component and not all vertices have been labelled yet,
  then there is no solution.\qed
\end{sketch}

\subsection{Tree-Child and Normal Networks}

\label{sec:tc}

A rooted network is \emph{tree-child} if each vertex~$u$ has at least one
child~$v$ that is a tree vertex or a leaf, implying a reticulation-free path
from each vertex to a leaf.
We say~$v$ is a \emph{tree child} of~$u$.
A network is \emph{normal} if it is tree-child and contains no arc~$uv$ such
that $\Net$ contains a path from $u$ to $v$ other than~$uv$.
We call such an arc $uv$ a \emph{shortcut}.
Since every tree-child network is an orchard, and every orchard is tree-based,
every tree-child network has a base tree.
By \cref{lem:blob-determined}, the following theorem follows if we can show
that it takes $\Oh(10.6667^\ell \cdot |B|)$ time to decide whether a blob $B$
has a tree-child or normal orientation with a given root $\rho_B$.

\begin{theorem}
  \label{thm:tc}
  \PROB{Tree-Child Orientation} and \PROB{Normal Orientation} can be solved in
  $\Oh(10.6667^\ell \cdot n^2)$~time.
\end{theorem}

To decide whether a blob $B$ has a tree-child or normal orientation with
root~$\rho_B$, we use the same strategy as for orchards, guessing an unrooted
spanning tree of $B$, which defines a unique directed spanning tree $S$ with
root $\rho_B$, and then deciding whether $B$ has a tree-child or normal
orientation consistent with $S$.
However, while, given the spanning tree $S$, this problem can be solved in
$\Oh(|B|)$ time for \PROB{Orchard Orientation}, doing the same for
\PROB{Tree-Child Orientation} and \PROB{Normal Orientation} takes $\Oh(2^\ell
\cdot |B|)$ time, as we show next.
Thus, the cost per blob $B$ is $\Oh(5.3334^\ell \cdot 2^\ell \cdot |B|)$, and
we obtain the bound in \cref{thm:tc}.

Recall that each side $e = \{u, v\}$ of $\gen{B}$ corresponds to a path $P_e$
in $B$.
Since each internal vertex of $P_e$ has an attached leaf, providing a
guaranteed tree child for this vertex, it suffices to ensure that every vertex
of $\gen{B}$ has a tree child.
Since the directions of the edges in a path $P_e$ corresponding to an arc $e
\in S$ are fixed, our aim is to orient the edges in $P_e$, for each side $e =
\{u,v\} \notin S$ so that $u$ and/or $v$ has a tree child in $P_e$, if
possible.
For any $\ge 3$-side $e \notin S$, we can ensure that both endpoints have a
tree child in $P_e$ by making the middle vertex of $P_e$ a reticulation and
directing all edges in $P_e$ towards this reticulation.
Thus, only the directions of the edges in paths corresponding to $\le 2$-sides
not in $S$ must be chosen carefully.
In the following, whenever we choose the orientation of the edges of a path
$P_e$ corresponding to a $2$-side, we make one of the internal vertices of
$P_e$ a reticulation; it is easily verified that this is always a correct
choice.

We start by guessing the directions of all $0$-sides and partition the $1$- and
$2$-sides into the following subsets:
Let $E_1$ be the set of these sides.
A \emph{singleton side} is one that does not share an endpoint with any other side
in $E_1$.
All other sides in $E_1$ form a collection of paths and cycles, which we call
\emph{leaf paths} and \emph{leaf cycles}, respectively, because their internal
vertices are leaves of $S$.

\begin{restatable}[\appendixstar]{lemma}{breakingpoint}\label{lem:breaking-point}
  If $B$ has a tree-child orientation consistent with $S$, then there is such a
  tree-child orientation $B^r$ such that exactly one of the following holds for
  each leaf path $\langle v_0, \dots, v_t \rangle$:
	\begin{itemize}
    \item $v_i$ has a tree child in~$P_{\{v_i,v_{i+1}\}}$, for all~$i \in
      [t-1]_0$.
    \item $v_i$ has a tree child in~$P_{\{v_i,v_{i-1}\}}$, for all~$i \in [t]$.
    \item There is an index $i \in [t - 1]_0$ such that:
    \begin{itemize}
      \item each vertex $v_j$ with $j \in [i]$ has a tree child in
        $P_{\{v_{j-1}, v_j\}}$,
      \item each vertex $v_j$ with $j \in \{i+1, \dots, t\}$ has a tree child
        in $P_{\{v_j, v_{j+1}\}}$,
      \item $P_{\{v_i, v_{i+1}\}}$ contains a reticulation, and
      \item if~$\{v_i, v_{i+1}\}$ is a $2$-side, then $v_i$ (or $v_1$ if $i=0$)
        has a tree child in~$P_{\{v_i, v_{i+1}\}}$.
    \end{itemize}
  \end{itemize}
\end{restatable}

\thmtoappendix{lem:breaking-point}{\breakingpoint*}{
\begin{proof}
  As only one of the conditions in the lemma can be satisfied by $B_r$, we show that at least one is
  satisfied. Fix a leaf path $Q = \langle v_0, \dots, v_t \rangle$. The lemma
  holds if $Q$ satisfies one of the first two conditions. So assume
  that there are vertices~$v_a$ and~$v_b$ such that $v_a$ has a tree child
  in~$P_{\{v_a,v_{a-1}\}}$ and $v_b$ has a tree child in~$P_{\{v_b,v_{b+1}\}}$.
  Choose $a$ to be the maximal such index, and $b$ to be the minimal such
  index. Since $Q$ contains no $\ge 3$-sides and, therefore, each side in $Q$
  can provide a tree child for at most one generator vertex, $v_i$ has to have a
  tree child in~$P_{\{v_i,v_{i+1}\}}$ for all~$i \geq b$, and $v_i$ has to have
  a tree child in~$P_{\{v_i,v_{i-1}\}}$ for all~$i \leq a$. This precludes that
  $a > b$. Thus, $a \le b$. If $a + 1 < b$, then $v_{a+1}$ must have a
  tree child in~$P_{\{v_a, v_{a+1}\}}$ or $P_{\{v_{a+1}, v_{a+2}\}}$,
  contradicting the choice of $a$ or $b$.  Thus, $a + 1 \le b$.  In this case,
  the index $i = a$ satisfies the first two conditions of the third case, and
  we can orient the edges in $P_{\{v_a, v_{a+1}\}}$ so that an internal vertex
  of $P_{\{v_a, v_{a+1}\}}$ becomes a reticulation, which satisfies the last
  two conditions of the third case.\qed
\end{proof}
}

\begin{restatable}[\appendixstar]{lemma}{leafcycles}\label{lem:leaf-cycles}
  If $B$ has a tree-child orientation consistent with $S$, then there is such a
  tree-child orientation $B^r$ such that exactly one of the following two statements
  holds for each leaf cycle $C = \langle v_0, \dots, v_t, v_0 \rangle$:
  \begin{itemize}
    \item $v_i$ has a tree child in~$P_{\{v_i,v_{(i+1) \bmod (t+1)}\}}$, for all~$i \in
      [t]_0$.
    \item $v_i$ has a tree child in~$P_{\{v_i,v_{(i-1) \bmod (t+1)}\}}$, for all~$i \in [t]_0$.
  \end{itemize}
\end{restatable}

\thmtoappendix{lem:leaf-cycles}{\leafcycles*}{
\begin{proof}
	Assume without loss of generality that~$v_0$ has a tree child in~$P_{\{v_0,v_1\}}$.
	Since $C$ contains no $\ge 3$-sides, the same inductive argument as in
	the proof of \cref{lem:breaking-point} shows that $v_j$ has a tree child
  in~$P_{\{v_j, v_{(j+1) \bmod (t+1)}\}}$ for all $j \in [t]$.
	The second option can be shown analogously.\qed
\end{proof}
}


For each leaf cycle or leaf path, we guess an orientation from among the ones
allowed by \cref{lem:breaking-point,lem:leaf-cycles}.
If $m_0$ is the number of $0$-sides, $m_1$ is the number of leaf cycles and
leaf paths, and $t_i$ is the number of sides in the $i$th leaf cycle or path,
then the total number of orientation guesses we make is bounded by~$2^{m_0} \cdot
\prod_{i=1}^{m_1} 2^{t_i + 2}$.
Since $S$ is a spanning tree of $B$ and $B$ has level~$\ell$, we have $m_0 +
\sum_{i=1}^{m_1} t_i \le \ell$.
Since also $t_i \ge 2$, for all $i \in [m_1]$, it now follows immediately that
our algorithm inspects at most $2^\ell$ combinations of guesses of orientations
of $0$-sides, leaf paths, and leaf cycles.
It remains to describe how to choose the orientations of singleton sides in
$\Oh(|B|)$ time.

First, we test whether the edge directions chosen so far create a directed
cycle, which takes $\Oh(|B|)$ time.
If so, there is no orientation consistent with these guesses, so we
try the next guess.
Otherwise, we apply the following reduction rule and lemma.

\begin{rrule}
	\label{rr:orient-edges-tc}
  Let $e = \{u, w\}$ be a singleton side.  The other neighbours of $u$ and $w$
  in $\gen{B}$ are their parents in $S$ as well as two vertices $c_u$ and $c_w$ such
  that $c_u$ is the child of $u$ in $S$ or $\{u, c_u\}$ is a $\ge 3$-side, and
  analogously for $c_w$.
  If $u$ has no tree child in
  $P_{\{u,c_u\}}$ or $w$ has no tree child in $P_{\{w,c_w\}}$, assume w.l.o.g.\
  that the former is the case.  Then orient $P_e$ so that $u$'s neighbour in
  $P_e$ is $u$'s tree child and, if $e$ is a $2$-side, the other internal
  vertex of $P_e$ is a reticulation.
\end{rrule}

Once \cref{rr:orient-edges-tc} can no longer be applied, we test once more
whether the choices made so far create a directed cycle and, if so, reject the
current guess of orientations of $0$-sides, leaf cycles, and leaf paths.
Otherwise, we orient the remaining sides by making one of the internal vertices
of each of these sides a reticulation.

\begin{restatable}[\appendixstar]{lemma}{allremainingedges}\label{lem:reticulating-all-remaining-edges}
  When all leaf paths, leaf cycles, and $0$-sides have been oriented and
  \cref{rr:orient-edges-tc} cannot be applied, all remaining sides can be
  oriented so that one of the internal vertices of each such side becomes a
  reticulation.
\end{restatable}

\thmtoappendix{lem:reticulating-all-remaining-edges}{\allremainingedges*}{
\begin{proof}
  Consider any side $e = \{u,w\}$ that remains unoriented after applying
  \cref{rr:orient-edges-tc} exhaustively.
  Let $c_u$ and $c_w$ be as chosen in the rule.
  Then the neighbour of $u$ in $P_{\{u,c_u\}}$ is a tree child of $u$, or it is
  a vertex of $\gen{B}$ incident to a still unoriented singleton side.
  The same is true for the neighbour of $w$ in~$P_{\{w,c_w\}}$.
  Placing a reticulation on an internal vertex of every still unoriented
  singleton side ensures that the endpoints of each such side become
  tree vertices.
  Thus, the neighbour of $u$ in $P_{\{u,c_u\}}$ and the neighbour of $w$ in
  $P_{\{w,c_w\}}$ are tree children of $u$ and $w$, respectively, once we are
  done orienting all edges, and there is no need to choose the edge directions in
  $P_e$ so that $u$ or $w$ has a tree child in $P_e$.
  Since the edges of each side not oriented after applying
  \cref{rr:orient-edges-tc} are oriented in opposite directions, we also do not
  introduce any cycles.\qed
\end{proof}
}

Once all edges have been oriented, we verify that the network is tree-child,
which is easily done in $\Oh(|B|)$ time.
The running time of the algorithm now follows: As already argued, it suffices
to show that the cost per guess of spanning tree and orientations of $0$-sides,
leaf cycles, and leaf paths is $\Oh(|B|)$, as there are~$\Oh(5.3334^\ell)$~spanning
trees and at most $2^\ell$ guesses of orientations of $0$-sides, leaf
cycles, and leaf paths.
Testing for cycles and the final orientation of singleton sides after
\cref{rr:orient-edges-tc} is no longer applicable clearly take $\Oh(|B|)$
time.
The application of \cref{rr:orient-edges-tc} can be implemented by traversing
$S$ bottom-up and testing the status of~$c_u$ and~$c_w$ for every singleton
edge $\{u, w\}$.
Thus, this also takes $\Oh(|B|)$ time.

Since we explicitly test whether we obtain a tree-child orientation, to
establish the correctness of the algorithm, it suffices to show that our
algorithm finds a tree-child orientation of $B$ consistent with $S$ if such an
orientation exists.
By \cref{lem:breaking-point,lem:leaf-cycles}, we consider all necessary guesses
of orientations of $0$-sides, leaf cycles, and leaf paths.
\Cref{rr:orient-edges-tc} chooses edge orientations forced by the current
guess.
By \cref{lem:reticulating-all-remaining-edges}, the orientations chosen for the
remaining edges have no impact on whether the produced orientation is
tree-child.
Thus, our algorithm finds a tree-child orientation consistent with $S$ if one
exists.

To obtain an orientation into a normal network, we apply an additional rule
after \cref{rr:orient-edges-tc} and before applying
\cref{lem:reticulating-all-remaining-edges}:

\begin{rrule}
	\label{rr:orient-edges-norm}
  If there is a singleton side $e = \{u,w\}$ such that the current partial
  orientation contains a directed path from~$u$ to~$w$, then, (1) if~$k(e)=1$,
  orient the edges in $P_e$ towards $w$ and, (2) if~$k(e)=2$, orient the edges
  in $P_e$ so that the neighbour of $w$ in $P_e$ is a reticulation.
  Both choices prevent shortcuts.
\end{rrule}

\section{Finding Orientations by Guessing the Reticulations}

\label{sec:reti-guessing}

For the network classes considered in this section, it is more efficient to
guess the placement of reticulations than guessing a spanning tree.
Specifically, to decide whether a blob $B$ has a $\C$-orientation with root
$\rho_B$, we guess which of the vertices of~$\gen{B}$ are reticulations and
which sides of $\gen{B}$ have reticulations on them.
Then we show that, given this guess of reticulations, it takes polynomial time
to decide whether there is a $\C$-orientation of $B$ consistent with this
guess.

\begin{restatable}[\appendixstar]{lemma}{eqclasses}\label{lem:num-equivalence-classes}
  For an unrooted generator $\gen{B}$ of a level-$\ell$ blob, there
  are~$\Oh(12.2071^\ell)$ possible choices for placing reticulations.
\end{restatable}

\thmtoappendix{lem:num-equivalence-classes}{\eqclasses*}{
\begin{proof}
  By Lemma~1 in \cite{kelk-linz-tbr}, $\gen{B}$ has $2\ell - 2$ vertices and $3\ell - 3$
  edges, yielding $5\ell - 5$ reticulation candidates.
  Since we need to place $\ell$ reticulations, we thus have
  $\binom{5\ell-5}{\ell}$ choices to do so.
  Using the entropy bound $\binom{5\ell}{\ell} \leq (5^5 / 4^4)^\ell = (3125 /
  256)^\ell$, this yields $\Oh(12.2071^\ell)$ choices.\qed
\end{proof}
}

\looseness=-1
Let $V_r \subseteq V(\gen{B})$ be the subset of generator vertices we guessed
are reticulations, and let $E_r \subseteq E(\gen{B})$ be the subset of sides
that we guessed have reticulations on them.
There is at most one orientation of the subgraph of $\gen{B}$ containing
all sides of $\gen{B}$ not in $E_r$ in which exactly the vertices in $V_r$ are
reticulations, and this orientation can be found
in~$\Oh(|\gen{B}|)$~time if it exists~\cite{huber2024orienting}.
The orientation of each side $e \notin E_r$ determines the orientation of the
edges in the corresponding path $P_e$ in $B$.
Thus, we only need to find the orientations of the edges in paths~$P_e$
for~$e
\in E_r$.
For all of the network classes we discuss here, it is always safe to place the
reticulation on a $\ge 3$-side $e \in E_r$ at any vertex of $P_e$ that is not a
neighbour of either endpoint of $P_e$.
$0$-sides cannot have any reticulations on them.
If a $1$-side $e$ has a reticulation on it, then this determines the
orientation of all edges in $P_e$.
Thus, it only remains to choose the orientations of $2$-sides in~$E_r$.
We call these sides \emph{undetermined} and refer to the set of these sides
as~$U$.

\subsection{Temporal Networks}

\label{sec:temp}

A rooted network $\Net$ is \emph{temporal} if it has a labelling $\lambda:
V(\Net) \to \mathbb{Q}$ such that
\begin{enumerate}[(1)]
  \item each arc~$uv$ satisfies~$\lambda(u) \leq \lambda(v)$,
  \item each reticulation arc~$uv$ (i.e., $v$ is a reticulation)
    satisfies~$\lambda(u) = \lambda(v)$, and
  \item each non-leaf vertex $v \in V(\Net)$ has an arc~$vw$ such
    that~$\lambda(v) < \lambda(w)$.
\end{enumerate}
The vertex $w$ in (3) is a tree child of~$v$, so temporal networks
are tree-child.

By \cref{lem:blob-determined}, the following theorem follows if we can show how
to find a temporal orientation of a blob $B$ with a given root $\rho_B$ in
$\Oh(12.2071^\ell \cdot |B|)$ time.
By \cref{lem:num-equivalence-classes}, this in turn follows if we can decide in
$\Oh(|B|)$ time whether there is such an orientation with root $\rho_B$ and
consistent with a given guess of the two sets $V_r$ and $E_r$ defined at the
beginning of this section.

\begin{theorem}
  \label{thm:temporal}
  \PROB{Temporal Orientation} can be solved in $\Oh(12.2071^\ell \cdot n^2)$
  time.
\end{theorem}

To decide whether a blob $B$ has a temporal orientation with root $\rho_B$ and
consistent with $V_r$ and $E_r$, we orient all but the undetermined sides of
$\gen{B}$ as described at the beginning of this section, which takes $\Oh(|B|)$
time.
Then we try to extend this partial orientation $\N$ to a tree-child orientation
of $B$, as described below.
If no such orientation exists, then there is no temporal orientation consistent
with $\rho_B$, $V_r$, and $E_r$ either.
If we find a tree-child orientation, then the procedure below fixes the
orientations of all undetermined sides that matter to obtain a tree-child
orientation.
Any orientation of the remaining sides produces a tree-child orientation, but we
want a \emph{temporal} orientation.
We reduce the problem of orienting the remaining sides to \TTsOshort, the
problem defined in \cref{sec:orchards}.

To decide whether $\N$ can be extended to a tree-child orientation, note that
every generator vertex has at most two incident undetermined sides:
otherwise, it is not reachable from the root in any orientation consistent with
$E_r$, as every side incident to it contains a reticulation.
We start by partially orienting each undetermined side $e = \{u, v\}$ by
orienting the edges in $P_e$ incident with~$u$ and~$v$ towards the internal
vertices of $P_e$, as this is consistent with having a reticulation in $P_e$.
Now we apply the following reduction rule exhaustively, which is easily seen to
take $\Oh(|B|)$ time.

\begin{rrule}
  \label{rr:force-edges}
  If there is a generator vertex $u$ incident with exactly one undetermined
  side $e = \{u,v\}$ and without a tree child along any of its other incident
  sides, then orient the edges in $P_e$ so that the internal vertex of $P_e$
  next to $u$ is a tree vertex, and the other internal vertex is a
  reticulation.
\end{rrule}

After applying this rule exhaustively, we check whether the resulting partial
orientation is acyclic and whether every vertex now has a tree child.
If not, then there is no tree-child orientation consistent with the guessed
placement of reticulations and, thus, also no temporal orientation.

To orient the remaining sides to obtain a temporal orientation, or decide that
none exists, we construct and solve the following \TTsOshort instance in
$\Oh(|B|)$ time.
The set~$A$ contains an arc~$uw$ for each side~$e=\{u,w\} \notin E_r$
that is directed from~$u$ to~$w$ unless $e$ is a $0$-side and
$w$ is a reticulation.
It also contains an arc $uw$ for every side $e = \{u, w\} \in U$
for which \cref{rr:force-edges} placed the reticulation in $P_e$ next to $w$.
The set~$E_=$ contains all $1$-sides in $E_r$ and all $0$-sides
incident with a reticulation.
The set~$E_{\neq}$ contains all undetermined sides as well as~$\{u,w\}$ if~$uw
\in A$.
The set $V_2$ contains all generator vertices incident with a~$\geq 3$-side.
All remaining generator vertices are in~$V_1$.

By the following lemma, $B$ has a temporal orientation with root $\rho_B$ and
consistent with $V_r$ and $E_r$ if and only if this \TTsOshort instance has a
solution.
As all steps of the algorithm, for a fixed choice of $\rho_B$, $V_r$, and
$E_r$, take $\Oh(|B|)$ time, this proves \cref{thm:temporal}.

\begin{lemma}
  There is a temporal orientation of $B$ consistent with $\rho_B$, $V_r$,
  and~$E_r$ if and only if the \TTsOshort instance just defined has a solution.
\end{lemma}

\begin{proof}
  First, assume that the \TTsOshort instance has a solution $\lambda$.
  We extend $\lambda$ to a temporal labelling of $B$ as follows:
  For each side~$e = \{u,w\}$ with corresponding path~$P_e = \langle u, v_1,
  \dots, v_{k(e)}, w \rangle$ in $B$, assume w.l.o.g.\ that~$\lambda(u) \leq
  \lambda(w)$.
  If $e \notin E_r$, then set~$\lambda(v_i) = \lambda(u) +
  \frac{i}{k(e)+1}$, for all $i \in [k(e) - 1]$; if $w$ is not a reticulation,
  then set $\lambda(v_{k(e)}) = \lambda(u) + \frac{k(e)}{k(e)+1}$; if~$w$ is a
  reticulation, then set~$\lambda(v_{k(e)}) = \lambda(w)$.
  If $e \in E_r$ is a $1$-side, then set~$\lambda(v_1) = \lambda(u)$.
  If $e \in E_r$ is a $2$-side, then set~$\lambda(v_1) =
  \lambda(v_2) = \lambda(w)$.
  If $e \in E_r$ is a $\geq 3$-side, then
  set~$\lambda(v_i) = \lambda(w) + \min(i,k(e)-2)$.
  For any leaf $x$ with parent $p$, set $\lambda(x) = \lambda(p) + 1$.
  This labelling defines an orientation by directing an edge~$\{u,v\}$
  towards~$v$ whenever~$\lambda(u) < \lambda(v)$ and orienting~$\{u,v\},
  \{v,w\}$ towards~$v$ whenever~$\lambda(u) = \lambda(v) = \lambda(w)$.
  It is is easy to verify that $\lambda$ is a temporal labelling of $B$, so the
  orientation is temporal.
  In particular, we explicitly ensure that the endpoints of reticulation arcs
  have the same label and that every vertex $v$ has an out-neighbour $w$ with
  $\lambda(w) > \lambda(v)$.

  Now, assume that $B$ has a temporal orientation and a corresponding
  temporal labelling $\lambda$ consistent with $\rho_B$, $V_r$, and $E_r$.
  Since this choice of reticulations allows only one orientation of the sides
  without reticulations on them, the temporal orientation agrees with the
  orientations of these sides computed by the algorithm from
  \cite{huber2024orienting}.
  It must also agree with the orientations chosen by \cref{rr:force-edges}, as
  these orientations are necessary even to obtain a tree-child orientation.
  As the endpoints of reticulation arcs must be assigned the same label by
  $\lambda$ and every vertex must have a neighbour with a greater label, it
  is now easy to verify that the restriction of $\lambda$ to the generator
  vertices is a solution of the \TTsOshort instance.\qed
\end{proof}


\ifConference

\subsection{Other Network Classes}

A network is \emph{reticulation-visible} if, for each reticulation~$r$, there
is a leaf~$x$ such that every path from the root to~$x$ contains~$r$.
We say $r$ is \emph{visible} from $x$.
A~network is \emph{tree-sibling} if each reticulation has a sibling that is a
tree vertex, or equivalently, if one parent of each reticulation has a
tree child.

To solve the \PROB{Reticulation-Visible} and \PROB{Tree-Sibling Orientation}
problems, we use an approach similar to the solution of \PROB{Temporal
Orientation}: For each blob $B$, we guess the placement of reticulations on the
vertices and sides of $\gen{B}$ and then decide in linear time whether there
is a reticulation-visible or tree-sibling orientation consistent with each
guess.
Due to lack of space, the details are deferred to the full paper.

\begin{theorem}
  \label{thm:rv-tree-sibling}
  \PROB{Reticulation-Visible} and \PROB{Tree-Sibling Orientation} can be solved
  in $\Oh(12.2071^\ell \cdot n^2)$ time.
\end{theorem}

\else
\subsection{Reticulation-Visible Networks}

\label{sec:rv}

Similar to \PROB{Temporal Orientation}, by
\cref{lem:blob-determined,lem:num-equivalence-classes}, the following theorem
follows if we can decide in $\Oh(|B|)$ time whether a blob $B$ has a
reticulation-visible orientation with a given root $\rho_B$ and such that the given
sets $V_r$ and $E_r$ are exactly the sets of vertices and sides of $\gen{B}$
that are/contain reticulations of the orientation.

\begin{theorem}
	\PROB{Reticulation-Visible Orientation} can be solved in\\$\Oh(12.2071^\ell
	\cdot n^2)$ time.
\end{theorem}

As for \PROB{Temporal Orientation}, we start by testing whether the subgraph of
$\gen{B}$ excluding the edges in $E_r$ has an orientation whose reticulations
are exactly the vertics in $V_r$.
This determines the orientations of the edges in $P_e$, for all $e \notin E_r$.
We also orient all $1$- and $\ge 3$-sides as discussed at the beginning of
\cref{sec:reti-guessing}.
It remains to orient the undetermined sides in $E_r$, that is, the sides
in~$U$.
Reticulations on these sides are trivially visible from their adjacent leaves.
Thus, when orienting the sides in $U$, our only concern is to ensure that all
reticulations in $V_r$ are visible.
We reduce this to \MatchBG.

Note that the edges in $P_e$, for any $e \in U$, cannot be part of any path
from the root of $B$ to a leaf not adjacent to a vertex in $P_e$, because the
first and last edge of $P_e$ are directed towards the internal vertices of
$P_e$.
Thus, if based on the orientations of the sides not in $U$, a reticulation in
$V_r$ is visible from some leaf, it will remain visible no matter the
orientations we choose for the sides in~$U$.
We remove all these visible reticulations from $V_r$.
Similarly, if based on the orientations of the sides not in~$U$, a
reticulation $r \in V_r$ is visible from another reticulation $r' \in V_r$,
then $r$ is visible in the final orientation if $r'$ is visible in this
orientation.
Thus, we can also remove those reticulations from $V_r$ that are visible from
other reticulations in~$V_r$.
Each of the reticulations that remain in $V_r$ must be made visible from a leaf on a
side in~$U$.
This is the problem we reduce to \MatchBG.

Before describing this reduction, we discuss how to identify the reticulations that can be removed
from~$V_r$.
Let $B'$ be the graph obtained from $B$ by replacing every undetermined side
with a $1$-side with a reticulation on it.
Then we run directed BFS in $B'$ from each vertex $r \in V_r$, halting along
each branch upon encountering a leaf or a reticulation.
We remove $r$ from $V_r$ if this BFS finds a leaf or it discovers some
reticulation $r'$ twice in different branches of the search.
By the following lemma, this procedure correctly identifies all reticulations
that can be removed.

\begin{lemma}
	\label{lem:filter-correct}
  A reticulation~$r$ is visible from a leaf or from another reticulation if and
  only if the BFS from $r$ finds a leaf or it finds a reticulation twice.
\end{lemma}

\begin{proof}
  First, observe that, if a leaf $\ell$ is found, $r$~has a path to $\ell$ that
  does not pass through any reticulation; thus, $r$ is visible from $\ell$.
  If a reticulation $r'$ is encountered along two distinct branches, $r$~is
  visible from $r'$.

  Now assume that the BFS does not discover any leaves and discovers every
  reticulation at most once.
  Then each branch of the BFS ends at a unique reticulation.
  Thus, for each of these reticulations there is a path from the root
  avoiding~$r$.
  This directly implies that every vertex below these reticulations has a path
  from the root avoiding~$r$ as well, and $r$ cannot be visible from any leaf
  or any other reticulation.\qed
\end{proof}

Having identified the reticulations in $V_r$ not yet visible from a leaf nor
from another reticulation, we need to orient the undetermined sides to make
these reticulations visible from leaves attached to these sides.
We construct a bipartite graph~$H$ with vertex set $V_r \uplus U$ and with an
edge between $r \in V_r$ and $\{u,v\} \in U$ if and only if, in $B$, $r$ can
reach $u$ or $v$ without passing through another reticulation.
Then we find a maximum matching $M$ in $H$.
By the following lemma, $B$ has a reticulation-visible orientation with root
$\rho_B$ and consistent with $V_r$ and $E_r$ if and only if $|M| = |V_r|$, and
the proof shows how to construct this orientation from $M$.

\begin{lemma}
  \label{lem:rv-matching}
  $B$ has a reticulation-visible orientation with root $\rho_B$ and consistent
  with $V_r$ and $E_r$ if and only if $|M| = |V_r|$.
\end{lemma}

\begin{proof}
  A matching $M$ of size $|V_r|$ in $H$ yields a valid orientation as follows:
  If $(r,\{u,v\}) \in M$, then without loss of generality, $r$ has a path to
  $u$ that does not pass through another reticulation.
  We choose the reticulation in $P_{\{u,v\}}$ to be the neighbour of $v$, which
  ensures that $r$ is visible from the leaf attached to the neighbour of $u$ in
  $P_{\{u,v\}}$.
  For each undetermined side not matched by $M$, we place the reticulation on
  this side arbitrarily.

  Conversely, assume that there is a reticulation-visible orientation of $B$
  with root $\rho_B$ and consistent with $V_r$ and $E_r$.
  Then every reticulation $r \in V_r$ is visible from some leaf $\ell$.
  This leaf must be attached to some side $\{u, v\} \in U$ and, w.l.o.g., there
  is a path from $r$ to $u$ that does not pass through any other
  reticulation, because $r$ is not visible from a leaf or from any other
  reticulation based on the partial orientation before orienting the sides
  in~$U$.  Thus, $(r,\{u,v\})$ is an edge of~$H$, and we can add it to~$M$.
  Since each undetermined side has a reticulation on it, only one endpoint of
  such a side is visible from a leaf attached to this side.  This immediately
  implies that $M$ is a matching, and its size is $|V_r|$ by definition.\qed
\end{proof}

\Cref{lem:rv-matching} establishes the correctness of our algorithm for
\PROB{Reticulation-Visible Orientation}.
It remains to bound the cost of deciding whether a blob $B$ has an orientation
with root $\rho_B$ and consistent with $V_r$ and $E_r$.
As already discussed, it takes $\Oh(|B|)$ time to test whether the subgraph of
$B$ excluding the edges in $E_r$ has an orientation consistent with $V_r$.
The cost of running BFS to identify those vertices that can be removed from
$V_r$ because they are visible from a leaf or another reticulation is also
$\Oh(|B|)$: Indeed, while we run a separate search from every vertex in $V_r$,
the search from each vertex stops when it encounters a reticulation.
As a result, every vertex in $B$ is visited at most twice across all searches
from the vertices in $V_r$.
Finally, the construction of $H$ from $B$ clearly takes $\Oh(|B|)$ time, again
using directed BFS.
Given $H$, the matching $M$ can be found in $\Oh(|B|)$ time because every
vertex in $U$ has degree at most $2$ in $H$~\cite{schrijver}.
Producing the final orientation as discussed in the proof of
\cref{lem:rv-matching} is then easily verified to take $\Oh(|B|)$ time.
Thus, the cost of deciding whether a blob $B$ has an orientation with root
$\rho_B$ and consistent with $V_r$ and $E_r$ is $\Oh(|B|)$, as required.


\subsection{Tree-Sibling Networks}

\label{sec:ts}

Once again, we focus on testing in $\Oh(|B|)$ time whether a blob $B$ has a
tree-sibling orientation with the chosen root $\rho_B$ and consistent with the
guessed sets $V_r$ and $E_r$ of vertices and sides $\gen{B}$, respectively,
that are/contain reticulations.
The following theorem then follows from
\cref{lem:blob-determined,lem:num-equivalence-classes}.

\begin{theorem}
  \label{thm:ts}
  \PROB{Tree-Sibling Orientation} can be solved in $\Oh(12.2071^\ell \cdot
  n^2)$ time.
\end{theorem}

Once again, we start by orienting all sides not in $U$ as discussed at the
beginning of \cref{sec:reti-guessing}.
Next, note that any reticulation we place on an undetermined side of $B$ has a
parent with a tree child, as this side is a $2$-side and the other internal
vertex of the side becomes a tree vertex that is a parent of the reticulation
and whose other child is a leaf.
Thus, we only need to concern ourselves with ensuring that every reticulation
that is in $V_r$ or is an internal vertex of a $1$-side has a tree sibling.
Let $R$ be the set of these reticulations.

For any reticulation $r \in R$, let $u$ and $w$ be its parents, let $u'$ be the
other child of $u$, and let $w'$ be the other child of $w$.
If neither $u'$ nor $w'$ is an internal vertex of an undetermined side, then
their status as tree vertices or reticulations is already fixed, and we can
check whether at least one of them is a tree vertex.
If neither is a tree vertex, then there is no tree-sibling orientation
with root $\rho_B$ and consistent with $V_r$ and $E_r$.
Otherwise, $r$ has a tree sibling, and we remove $r$ from $R$.

The vertices that remain in $R$ are the ones that need to find a tree sibling
on one of the undetermined sides.
We reduce the problem of orienting the undetermined sides to ensure that each
vertex in $R$ has a tree sibling to \mbox{\TSAT}.
We associate a Boolean variable $x_z$ with every internal
vertex $z$ of an undetermined side.
This variable is true if we make $z$ a tree vertex.
Thus, for the two internal vertices $y$ and $z$ of each undetermined side, we
impose the condition that \mbox{$x_y = \neg x_z$}.

For every reticulation $r \in R$, at least one of $u'$ and $w'$ is an internal
vertex of an undetermined side.
If both are internal vertices of an undertermined side, then we add the clause
$\{x_{u'}, x_{w'}\}$ to the \TSAT formula, because at least one of $u'$ and
$w'$ must be chosen as the tree sibling of~$r$.
Otherwise, assume w.l.o.g.\ that $u'$ is an internal vertex of an undetermined
side, and $w'$ is not.
In this case, we add the clause $\{x_{u'}\}$ to the \TSAT formula, as only
$u'$ can be a tree sibling of~$r$, that is, it must be made a tree vertex.
We immediately obtain the following lemma.

\begin{lemma}
  \label{lem:ts-tsat}
  $B$ has a tree-sibling orientation with root $\rho_B$ and consistent with
  $V_r$ and $E_r$ if and only if the \TSAT formula just constructed has a
  satisfying truth assignment.
\end{lemma}

\Cref{lem:ts-tsat} establishes the correctness of our algorithm for
\PROB{Tree-Sibling Orientation}.
Testing whether the subgraph of $B$ excluding the edges in $E_r$ has an
orientation consistent with $V_r$ takes $\Oh(|B|)$ time.
Inspecting the siblings of every reticulation in $R$ to test whether each such
reticulation either has a tree sibling or a sibling that is an internal vertex
of an undetermined side, and constructing the \TSAT instance, takes $\Oh(|B|)$
time.
The resulting \TSAT instance can also be solved in $\Oh(|B|)$ time.
This finishes the proof of \cref{thm:ts}.


\section{Networks of Bounded Scanwidth}
\fi

As a final remark, we consider computing an orientation of minimum scanwidth,
as defined in~\cite{berry2020scanning}.
The scanwidth of a graph is the maximum scanwidth of its blobs.
The scanwidth of a blob is completely determined by the placement of
reticulations on the vertices and sides of its generator, independent of
the exact position of the reticulation on each side that contains a reticulation.
Computing the scanwidth of a graph of order $n$ takes $\tilde\Oh(4^n)$ time~\cite{holtgrefeScanwidth}.
The generator of a blob $B$ has order at most $2\ell - 2$, and we place up to
$\ell$ reticulations on its sides.
Thus, the scanwidth of the graph obtained after placing reticulations on the
vertices and sides of $\gen{B}$ can be computed in $\tilde\Oh(4^{3\ell})$ time.
Using \cref{lem:blob-determined}, this proves the following theorem.

\begin{theorem}
	\label{thm:scanwidth}
	For an unrooted network $\U$, an orientation of $\U$ of minimum scanwidth can
	be computed in $\tilde\Oh(781.25^\ell \cdot n^2)$ time.
\end{theorem}

\section{Discussion}

\label{sec:discussion}

The main contribution of our paper is to show that it is not necessary to guess
the exact placement of reticulations as done in \cite{huber2024orienting} to
solve the $\C$-orientation problem for a number of prominent network classes;
guessing the spanning tree of the generator corresponding to a base tree of the orientation or
which sides of the generator contain reticulations is enough, which leads to
faster running times than achieved in \cite{huber2024orienting}.
Yet we believe that further improvements and extensions to other network
classes are possible.
Most notably, we believe that understanding the structure of $0$-sides can be
helpful, as having many leaves on sides simplifies the problem.
Another interesting question is whether \PROB{Tree-Child} and
\PROB{Normal Orientation} can be solved in polynomial time after guessing the
spanning tree, as is possible for \PROB{Tree-Based} and
\PROB{Orchard Orientation}.
While orienting into tree-child, orchard, and tree-based networks is \NP-hard~\cite{dempsey2024wild,docker2025existence,huber2024orienting,iersel2026class},
\NP-hardness for reticulation-visible, temporal, tree-sibling, or galled
networks remains open.

\todosi{Add question about count of orientations into tree-child networks.}

\bibliographystyle{splncs04}
\bibliography{ref}

\newpage
\appendixproofs

\end{document}

%% file: commands.tex
\usepackage[T1]{fontenc}
\usepackage[utf8]{inputenc}
\usepackage{microtype}
\usepackage{amsmath}
\usepackage{amssymb}
\usepackage{enumerate}
\usepackage{mathtools}
\usepackage{graphicx}
\usepackage{multirow}
\usepackage[table]{xcolor}
\usepackage{tikz}
\usetikzlibrary {arrows.meta} 
\usetikzlibrary{graphs}
\usepackage{placeins}
\usepackage{nicefrac}
\usepackage{tabularx}
\usepackage{xspace}
\usepackage{algorithm}
\usepackage{algpseudocode}
\usepackage{etoolbox}
\usepackage{paralist}
\usepackage{booktabs}
\usepackage{lineno}

\ifFinal
\usepackage[disable]{todonotes}
\nolinenumbers
\else
\usepackage[notref, notcite]{showkeys}
\usepackage{todonotes}
\linenumbers
\fi

\usepackage{hyperref}
\hypersetup{
	colorlinks,
	linkcolor={red!50!black},
	citecolor={blue!90!black!80},
	urlcolor={blue!80!black}
}
\usepackage[capitalize,nameinlink]{cleveref}

\newcommand{\appendixproofs}{}
\newcommand{\toappendix}[1]{\gappto{\appendixproofs}{#1}}
\toappendix{
	
	\section{Omitted Proofs}
	\label{adx:proofs}
  In this appendix, we present proofs that were omitted from
  the paper due to space constraints.
}
\newcommand{\appendixalgs}{}
\newcommand{\alginappendix}[1]{\gappto{\appendixalgs}{#1}}
\alginappendix{
  \section{Pseudocode}
	\label{adx:algs}
  In this appendix, we present the algorithms as pseudocode.
  We note that all algorithms are fully explained in the paper and
  this appendix is solely for the benefit of readers who prefer pseudocode.
}

\newcommand{\thmtoappendix}[3]{\toappendix{
    \bgroup
		\subsection{Proof of~\Cref{#1}}
		#2
		
		#3
    \egroup
	}
}

\spnewtheorem{rrule}{Reduction Rule}{\itshape}{\rmfamily}
\spnewtheorem{construction}{Construction}{\bfseries}{\itshape}
\spnewtheorem{clm}{Claim}{\itshape}{\rmfamily}
\spnewtheorem{observation}[theorem]{Observation}{\bfseries}{\itshape}
\spnewtheorem*{sketch}{Proof sketch}{\itshape}{\rmfamily}

\usepackage{thmtools}

\newcommand{\kommentar}[1]{}
\newcommand{\Oh}{\ensuremath{\mathcal{O}}}

\newcommand{\proofpara}[1]{\smallskip
	
	\noindent\textit{#1.}}

\newcommand{\NP}{{\normalfont{NP}}\xspace}

\newcommand{\Net}{{\ensuremath{\mathcal{N}}}\xspace}

\newlength{\boxindent}
\newcommand{\problemdef}[3]{%
	\begin{trivlist}
		\item \leavevmode\tikz{\node [draw,inner sep=4.5pt] {\begin{minipage}{\textwidth-9.4pt}
					\raggedright
					\normalsize\textsc{#1}
					
					\smallskip
					
					\settowidth{\boxindent}{\textbf{Input:} }%
					\leavevmode\hangindent=\boxindent
					\textbf{Input:}~#2
					
					\leavevmode\hangindent=\boxindent
					\rlap{\textbf{Find:}}\phantom{\textbf{Input:}}~#3
			\end{minipage}};}
	\end{trivlist}
}

\newcommand{\PROB}[1]{{{\normalfont\textsc{#1}}}\xspace}

\newcommand{\TSAT}{\PROB{2-SAT}}
\newcommand{\MatchBG}{\PROB{Matching on Bipartite Graphs}}